%% file: CFSDM_ARXIV_3.tex
\documentclass[onecolumn,draftclsnofoot, 12pt, ]{IEEEtran}
\usepackage{times}
\usepackage[T1]{fontenc}
\usepackage{graphicx}
\usepackage{amsmath}
\usepackage{setspace}
\usepackage{amsbsy}
\usepackage{amsthm}
\usepackage{cite}
\usepackage{cite}
\usepackage{authblk}
\usepackage{subfigure}
\usepackage{amsfonts}
\usepackage{array}
\usepackage{epstopdf}
\usepackage{algorithm}
\usepackage{algorithmicx}
\usepackage{algpseudocode}
\usepackage{breqn}
\newtheorem{theorem}{Theorem}
\newtheorem{lemma}[theorem]{Lemma}

\begin{document}
\title{Downlink Precoding for Massive MIMO Systems Exploiting Virtual Channel Model Sparsity}
\author{Thomas~Ketseoglou,~\IEEEmembership{Senior~Member,~IEEE,} and Ender~Ayanoglu,~\IEEEmembership{Fellow,~IEEE}

\thanks{T. Ketseoglou is with the Electrical and Computer Engineering Department, California State Polytechnic University, Pomona, California (e-mail: tketseoglou@csupomona.edu). E. Ayanoglu is
with the Center for Pervasive Communications and
Computing, Department of Electrical Engineering and Computer Science,
University of California, Irvine (e-mail: ayanoglu@uci.edu). This work was partially supported by NSF grant 1547155.}
}


\maketitle
\begin{abstract}
In this paper, the problem of designing a forward link linear precoder for Massive Multiple-Input Multiple-Output (MIMO) systems in conjunction with Quadrature Amplitude Modulation (QAM) is addressed. First, we employ a novel and efficient methodology that allows for a sparse representation of multiple users and groups in a fashion similar to Joint Spatial Division and Multiplexing. Then, the method is generalized to include Orthogonal Frequency Division Multiplexing (OFDM) for frequency selective channels, resulting in Combined Frequency and Spatial Division and Multiplexing, a configuration that offers high flexibility in Massive MIMO systems.  A challenge in such system design is to consider finite alphabet inputs, especially with larger constellation sizes such as $M\geq 16$. The proposed methodology is next applied jointly with the complexity-reducing Per-Group Processing (PGP) technique, on a per user group basis, in conjunction with QAM modulation and in simulations, for constellation size up to $M=64$.
We show by numerical results that the precoders developed offer significantly better performance than the configuration with no precoder or the plain beamformer and with $M\geq 16$.
\end{abstract}
\IEEEpeerreviewmaketitle
\section{{INTRODUCTION}}
Massive MIMO  employs a very large number of antennas and enables very high spectral efficiency \cite{Marzetta1,Marzetta2,Marzetta3}. For Massive MIMO to be capable of offering its full benefits, accurate and instantaneous channel state information is required at the base station (BS). Within Massive MIMO research, the problem of designing an optimal linear precoder toward maximizing the mutual information between the input and output on the downlink in conjunction with a finite input alphabet modulation and multiple antennas per user has not been considered in the literature, due to its complexity. There are techniques proposed for downlink linear precoding in a multi-user MIMO scenario, e.g., Joint Spatial Division and Multiplexing (JSDM) \cite{JSDM1}, but their implementation has been challenging so far. In addition, there has been a lack of publications on how to realistically integrate OFDM in Massive MIMO with success and without sacrificing the spectral efficiency of the system. On the other hand, the problem of finite-alphabet input MIMO linear precoding has been extensively studied in the literature. Globally optimal linear precoding techniques were presented \cite{Xiao,Lamarca} for scenarios employing channel state information available at the transmitter (CSIT)\footnote{Under CSIT the transmitter has perfect knowledge of the MIMO channel realization at each transmission.} with finite-alphabet inputs, capable of achieving mutual information rates
much higher than the previously presented Mercury Waterfilling (MWF) \cite{Verdu1} techniques by introducing input symbol correlation through a unitary input
transformation matrix in conjunction with channel weight adjustment (power allocation). In addition, more recently, \cite{Max} has presented an iterative algorithm for precoder optimization for sum rate maximization of  Multiple Access Channels (MAC) with Kronecker MIMO channels. Furthermore, more recent work has shown that when only Statistical Channel State Information (SCSI)\footnote{SCSI pertains to the case in which the transmitter has knowledge of only the MIMO channel correlation matrices \cite{Khan,Weich} and the thermal noise variance.} is available at the transmitter, in asymptotic conditions when the number of transmitting and receiving antennas grows large, but with a constant transmitting to receiving antenna number ratio, one can design the optimal precoder by looking at an equivalent constant channel and its corresponding adjustments as per the pertinent theory \cite{SCSI}, and applying a modified expression for the corresponding ergodic mutual information evaluation over all channel realizations. This development allows for a precoder optimization under SCSI in a much easier way \cite{SCSI}. Finally, \cite{Lozano,TK_EA_QAM} present for the first time results for mutual information maximizing linear precoding with large size MIMO configurations and QAM constellations. Such systems are particularly difficult to analyze and design when the inputs are from a finite alphabet, especially with QAM constellation sizes, $M\geq 16$.

In this paper, we present optimal linear precoding techniques for Massive MIMO, suitable for QAM with constellation size $M\geq 16$ and CSIT. Two types of antenna arrays are considered for the Base Station (BS), Uniform Linear Arrays (ULA) and Uniform Planar Arrays (UPA). In the UPA case, we consider arrays deployed either over the $x,~y$ direction or the $z,~x$ one. We show that by projecting the per user antenna uplink channels on the DFT based angular domain, called virtual channel model (VCM) herein, a sparse representation is possible for the channels. Then, by dividing spatially ``distant''  users into separate spatial sectors, we show that the spatial virtual channel representations between these users become approximately orthogonal. We then show that the concept of JSDM \cite{JSDM1} can be easily applied in the sparse virtual channel model domain representation and show that linear precoding on the downlink using Per-Group Precoding (PGP) in conjunction with the Gauss-Hermite approximation in MIMO \cite{TE_TWC,TK_EA_QAM}. However, the issue of group-based decoding is still present at the destination. We employ two different methods toward mitigating this problem. Then, by generalizing the presented approach to the frequency-selective (FS) channel case and applying OFDM, we show that much more flexibility and gains are available by the techniques presented. We show that when OFDM is integrated in the VCM-based JSDM system developed, resulting in Combined Frequency and Spatial Division and Multiplexing (CFSDM), the system can offer much higher rates, overcome issues of spatial overlapping by employing different carriers between spatially overlapping groups, and also easily decode different users' data within a group. In all examples presented, we show high gains achievable by the proposed downlink precoding approach. More specifically, the paper makes the following contributions in Massive MIMO:
\begin{enumerate}
\item It provides an analytical framework that allows spatially separated user groups  to be approximately orthogonal and thus to require independent per-group precoding beams from the base station.
\item It proves that the presented semi-orthogonal decomposition fits the JSDM \cite{JSDM1} model.
\item It shows that the selected pre-beamforming matrices for the JSDM decomposition are optimal.
\item It employs both linear and planar arrays.
\item It generalizes the approach to include OFDM under frequency-selective channel conditions in a very flexible way.
\item It shows very significant gains in conjunction with PGP \cite{TK_EA_QAM} and QAM modulation.
\end{enumerate}

The paper is organized as follows: Section II presents the system model and problem statement. Then, in Section III, we present a novel virtual channel approach which allows for efficient downlink precoding in a JSDM fashion for ULA and UPA for narrowband channels. In Section IV we focus on FS channels, which naturally leads to OFDM type of systems imposed to the presented JSDM approach. In Section IV, we present numerical results for optimal downlink precoding on the system proposed that implements the Gauss-Hermite approximation in the block coordinate gradient ascent method in conjunction with the complexity reducing PGP methodology\cite{TK_EA_QAM}. Finally, our conclusions are presented in Section V.

\section{{SYSTEM MODEL AND PROBLEM STATEMENT}}
Consider the following uplink equation on a narrowband (flat-fading) Massive MIMO system with a single cell
\begin{equation}
{\mathbf y_u} = {\mathbf H_u}  {\mathbf x}_u + {\mathbf n_u},
\end{equation}
where
${\mathbf y_u}$ is the $N_u\times1$ received vector at the base, ${\mathbf H}_u = [{\mathbf H}_1,\cdots, {\mathbf H}_G]$ is the $N_u\times K_{eff}$ channel matrix of the received data from all $K$ users, employing $N_u$ receiving antennas at the base, with $K_{eff}= \sum_{g} N_{d,g}$, where $N_{d,g}$ is defined below, and where users have been divided into $G$ groups with $K_g$ users in group $g$ ($1\leq g \leq G$), with user $k$ of group $g$ denoted as $k^{(g)}$ and employing $N_{d,k^{(g)}}$ transmitting antennas, with ($\sum_{g=1}^G K_g = K$), ${\mathbf H}_g =[{\mathbf H}_{g^{(1)}}\cdots {\mathbf H}_{g^{(K_g)}} ]$ is group $g$'s uplink channel matrix of size $N_u\times N_{d,g}$, with $N_{d,g}$ comprising the total number of antennas in the group, i.e., $N_{d,g}=\sum_{k^{(g)}}N_{d,k^{(g)}}$, where ${\mathbf n}_u$ represents the independent, identically distributed (i.i.d.) complex circularly symmetric Gaussian noise of variance per component $\sigma_u^2 = \frac{1}{\mathrm{SNR}_{s,u}}$, where ${\mathrm{SNR}_{s,u}}$ is the channel symbol signal-to-noise ratio ($\mathrm {SNR}$). The uplink symbol vector of size ${\mathbf x}_u$ of size $\sum_{g}^G{N_{d,g}}\times 1$ has i.i.d. components drawn from a QAM constellation of order $M$.
The corresponding downlink equation can then be derived by using for the downlink channel matrix ${\mathbf H}_d={\mathbf H}_u^h$, assuming that Time Division Duplexing (TDD) is employed in the system, to be
\begin{equation}
{\mathbf y_d} = {\mathbf H}_u^h  {\mathbf x}_d + {\mathbf n_d},
\end{equation}
where ${\mathbf y}_d$ is the downlink received vector of size $\sum_{g=1}^G N_{d,g}\times 1$, ${\mathbf x}_d$ is the $N_u\times 1$ vector of transmitted symbols drawn independently from a QAM constellation, and the vector ${\mathbf n}_d$ of size $\sum_{g=1}^G N_{d,g}\times1$ is the downlink circularly symmetric complex Gaussian noise with independent components.
The optimal CSIT precoder ${\mathbf P}$ needs to satisfy
\begin{eqnarray}
\begin{aligned}
& \underset{\mathbf P}{\text{maximize}}
& & I({\mathbf x_d};{\mathbf y_d})\\
& \text{subject to}
& &  \mathrm {tr}({\mathbf P} {\mathbf P}^h) = N_u, \\ \label{eq_MMIMO}
\end{aligned}
\end{eqnarray}
where the constraint is due to keeping the total power emitted from the $N_u$ antennas constant.

The problem in (\ref{eq_MMIMO}) results in exponential complexity at both transmitter and receiver, and it becomes especially difficult for QAM with constellation size $M \geq 16$ or large MIMO configurations. There are two major difficulties in (\ref{eq_MMIMO}): a) There are $N_u$ input symbols in (\ref{eq_MMIMO}) where $N_u$ is very large, thus making the design of the precoder and its optimization practically impossible, and b) The decoding operation at the receiver needs to be performed by employing all elements of ${\mathbf y}_d$ simultaneously, another impossible demand due to the users being distributed over the entire cell. In order to circumvent these difficulties, the JSDM concept was proposed in \cite{JSDM1} which divides users into groups based on channel similarity. However, a major impediment to deploying JSDM in practice has been the lack of a simple way that identifies the different groups of users with ease. Furthermore, \cite{JSDM1} has employed Gaussian input symbols, an assumption that can lead to discrepancies as far as the precoder performance is concerned, especially in high $ \mathrm {SNR}$ \cite{Xiao,TK_EA_GLOBECOM}.  In this paper, a methodology that employs the virtual channel model, based on the DFT channel angular domain, is employed in order to facilitate the group selection problem in JSDM and then the methodology of \cite{TK_EA_QAM} is employed in order to allow for the design of an optimal overall precoder on a per group basis.

\section{THE NARROWBAND SYSTEM DESCRIPTION UNDER THE VIRTUAL CHANNEL MODEL}
\subsection{Uniform Linear Array (ULA) at the Base with Flat Fading}
We begin with a ULA deployed at the BS along the $z$ direction as depicted in Fig. 1 and for flat fading, i.e., $B< B_{COH}$, where $B,~B_{COH}$ are the RF signal bandwidth and the coherence bandwidth of the channel, respectively. Each user group on the uplink transmits from the same ``cluster'' of elevation angles $\theta_{g} \in [\theta_{g} -\Delta \theta, \theta_{g} +\Delta \theta]$, distributed uniformly in the support interval, thus each user's $k^{(g)}$ of group $g$, ( $1\leq k^{(g)} \leq K_g$ and $1\leq g \leq G$) transmitting antenna $n$ channel, ${\mathbf h}_{u,g,k,n} = \frac{1}{{\sqrt L}}\sum_{l=1}^L \beta_{lgkn} {\mathbf a}(\theta_{lgkn})$, where ${\mathbf a}(\theta_{lgkn}) = [1, \exp(-j{2\pi}D\cos(\theta_{lgkn})),\cdots, \exp(-j{2\pi D(N_u-1)}\cos(\theta_{lgkn}))]^t$ is the array response vector, with $D=d/\lambda$ representing the normalized distance of successive array elements, $\lambda$ being the wavelegth, $\theta_{lgkn}$ is the elevation (incidence) angle of the  $l$ path of group $g$ $k$ user's $n$ receiving antenna, and the path gains $\beta_{lgkn}$ are independent complex Gaussian random variables with zero mean and variance $1$, same for all users in the group. The VCM representation, presented in \cite{Sayeed}, is formed by projecting the original channel ${\mathbf H_u}$ to the $N_u$ dimensional space formed by the $N_u\times N_u$ DFT matrix $F_{N_u}$, with row $k$, column $l$ ($1\leq k,~l\leq N_u$) element equal to $\exp(-j\frac{2\pi}{N_u}(k-1)(l-1))$. For Massive MIMO systems, i.e., when $N_u \gg 1$, the following Lemma 1 and 2 as well as Theorem 1 are true.
\begin{figure}
\centering
\setcounter{figure}{0}
\includegraphics[height=3.6in,width=5.25in]{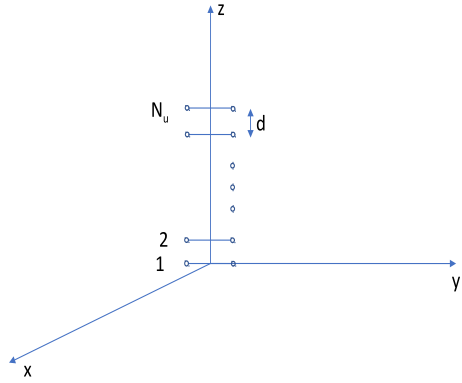}
	\caption{A ULA deployed across the $z$ axis, together with the projection of a tentative transmission point on the $x,~y$ plane.}
\end{figure}
\begin{lemma}
By employing the VCM for a ULA at the BS and under flat fading, the number of non-zero components of the VCM representation is small, i.e., the number of non-zero or significant elements in the channels of each group $g$ VCM representation, $|{\cal S}_g|,$ satisfy $|{\cal S}_g| \ll N_u$. Thus, in the VCM domain, a sparse overall group channel representation results.
\end{lemma}
\begin{proof}
By projecting each group channel ${\mathbf H}_{g}$ on the DFT virtual channel space \cite{Sayeed}, we get
\begin{equation}
{\tilde {\mathbf H}}_g = {\mathbf F}_{N_u}^h {\mathbf H}_g,
\end{equation}
where
${\mathbf F}_{N_u}$ is the DFT matrix of order $N_u$. Since each group attains the same angular behavior, over all users and antennas in the group, only a few, consecutive elements of ${\tilde {\mathbf H}}_g$ will be significant \cite{BEM}. This comes as a result of the fact that significant angular components need to be in the  main lobe of the response vector, i.e., the condition
 \begin{equation}
|\cos(\theta_{lgkn})-\frac{p}{DN_u}| \leq \frac{1}{DN_u}, \label{resolve1}
\end{equation}
   with $D = \frac{d}{\lambda}$, needs to be satisfied for angular component in the VCM $p$ ($1\leq p \leq N_u$) to be significant, i.e., with power $> 1$. From (\ref{resolve1}), we can easily see that the corresponding condition over the significant components becomes
   \begin{equation}
   DN_u\cos(\theta_{lgkn})-1 \leq p \leq DN_u\cos(\theta_{lgkn})+1,
   \end{equation}
   i.e., there are 3 significant non-zero components in the VCM representation for each channel's path. Since each path contains a different angle, due to the ULA model presented above, this number will be increased, but will be upper-bounded by $DN_u|\cos(\theta_{g}+\Delta \theta)-\cos(\theta_{g}-\Delta \theta) | +3 = 3 + 2DN_u|\sin(\theta_g)\sin(\Delta \theta)|\approx 3 + 2DN_u|\sin(\theta_g)|(\Delta \theta)$, where $\Delta \theta$ is in radians. For a typical scenario, $N_u =100$, $D=1/2$, $\theta_g = 30^\circ$, and $\Delta \theta = 4^\circ = 0.0698~\mathrm{radian}$, then the maximum number of non-zero (significant) paths is upper-bounded by 7.
\end{proof}
\begin{lemma}
Within the premise of the previous Lemma, if  $\cos(\theta_g-\Delta \theta ) < \cos(\theta_{g'}+\Delta \theta) -\frac{2}{DN_u}$, where $g$ and $g'$ represent two different groups ($g \neq g'$) and with $\theta_g > \theta_{g'}~ \text{and}~ 0\leq \theta_g,~ \theta_{g'} \leq 90^\circ$, then their support sets for each group are mutually exclusive, thus their corresponding virtual channel model beams (VCMB) become orthogonal. A similar relationship holds in the remaining quadrants.
\end{lemma}
\begin{proof}
When $\theta_g > \theta_{g'}~ \text{and}~ 0\leq \theta_g,~ \theta_{g'} \leq 90^\circ$, since the $\cos(\cdot)$ function is decreasing in this quadrant, we can easily see that the two support sets for the two groups, ${\cal S}_{g},~{\cal S}_{g'},$ will be disjoint. This comes from the fact that the assumed condition is equivalent to
\begin{equation}
\cos(\theta_g-\Delta \theta ) + \frac{1}{DN_u}< \cos(\theta_{g'}+\Delta \theta) -\frac{1}{DN_u},
\end{equation}
which means that the two support sets are not overlapping, by virtue of (\ref{resolve1}). We can develop similar conditions for all remaining quadrants.
Thus, by assuming adequate spatial separation between groups, we can ensure that the support sets of each group in the virtual channel representation do not overlap. Then, due to the non-overlapping of the support sets, there exists orthogonality between the components of each group in the virtual channel model, as it is next shown.
\end{proof}
\setcounter{theorem}{0}
\begin{theorem}
By employing the VCM for a ULA at the BS and under flat fading, provided user groups are sufficiently geographically apart, as per previous lemma, the channel model of the entire downlink channel can be expressed in a fashion that is fully suitable for JSDM type of processing where different groups become orthogonal and the downlink precoder is designed on a per group basis employing the virtual channel model representation alone. In the resulting JSDM type of decomposition, the corresponding group channel matrices are the virtual channel matrices of the group VCM projections and the group pre-beamforming matrices are the group's non-zero (significant) VCM beamforming directions.
\end{theorem}
\begin{proof}
By employing a size $|{\cal S}_g|\times N_u$ selection matrix\footnote{A selection matrix ${\mathbf S}^t$ of size $k\times n$ with $k<n$ consists of rows equal to different unit row vectors ${\mathbf e}_i$ where the row vector element $i$ is equal to  $1$ in the $i$th position and is equal to $0$ in all other positions. Such a matrix has the property that ${\mathbf S}^t {\mathbf S} = {\mathbf I}$.}
\begin{equation}
{{\mathbf H}_{g,v}}= {\mathbf S}_g^t{\tilde{\mathbf H}}_g={\mathbf S}_g^t{\mathbf F}_{N_u}^h {\mathbf H}_g,
\end{equation}
where the group $g$ virtual channel matrix is a reduced size, $r_g\times N_{d,g}$, matrix, with $r_g = |{\cal S}_{g}|$ the number of significant angular components in group $g$, due to the sparsity available in the angular domain. We can then write for the uplink group $g$ channel matrix ${\mathbf H}_g$,
\begin{equation}
{\mathbf H}_g = {\mathbf F}_{N_u}{\mathbf S}_g{\mathbf S}_g^t{\mathbf F}_{N_u}^h {\mathbf H}_g={\mathbf F}_{N_u, {\cal S}_g}{{\mathbf H}_{g,v}},
\end{equation}
where ${\mathbf F}_{N_u, {\cal S}_g}$ represents the selected columns of ${\mathbf F}_{N_u}$ due to its sparse representation in the angular domain. We can then write that due to non-overlapping supports in groups $g$, $g'$,  ${\cal S}_{g}\cap_{g\neq g'}{\cal S}_{g'} =\emptyset$, that
    \begin{equation}
    {\mathbf H}_{g}^h {\mathbf F}_{N_u, {\cal S}_g'} = {\mathbf 0},
    \end{equation}
    for $g\neq g'$. By TDD channel reciprocity, the group $g$ downlink channel matrix is given as
    \begin{equation}
    {\mathbf H}_{d,g} = {\mathbf H}_{g}^h = {{\mathbf H}_{g,v}}^h{\mathbf F}_{N_u, {\cal S}_g'}^h.
    \end{equation}
Since each group attains its non-zero virtual channel representation at non-overlapping positions, we can then use pre-beamformers provided by the matrix ${\mathbf B}=[{\mathbf F}_{N_u,~{\cal S}_1} \cdots {\mathbf F}_{N_u,~{\cal S}_G}]$. As we show below these pre-beamformers are optimal for the type of JSDM presented here.
Then, due to non-overlapping of the support sets, i.e., ${\cal S}_n\cap_{m\neq n}{\cal S}_m =\emptyset$, we see that the system becomes approximately orthogonal inter-group wise, i.e., $\sum_{m\neq g}{\mathbf H}_{\cal N}^h\cdot {\mathbf W}_{{\cal S}_g}^t{\mathbf W}_{{\cal S}_m}^{*} \approx 0$. Then,
\begin{equation}
\begin{split}
 {\mathbf y}_d =  \left[ \begin{array}{c} {{\mathbf H}_{1,v}}^h{\mathbf F}_{N_u, {\cal S}_1}^h \\
{{\mathbf H}_{2,v}}^h{\mathbf F}_{N_u, {\cal S}_2}^h \\
 \vdots\\
 {{\mathbf H}_{G,v}}^h{\mathbf F}_{N_u, {\cal S}_G}^h \end{array} \right]
 \left[ \begin{array}{c c c c} {\mathbf F}_{N_u, {\cal S}_1} &
{\mathbf F}_{N_u, {\cal S}_2} & \cdots &
{\mathbf F}_{N_u, {\cal S}_G} \end{array} \right]\\
\left[
\begin{array}{cccccc}
{\bf P}_1 & {\bf 0} & {\bf 0} & \cdots & {\bf 0} & {\bf 0}\\
{\bf 0} & {\bf P}_2 & {\bf 0} & \cdots & {\bf 0} & {\bf 0}\\
{\bf 0} & {\bf 0} & {\bf P}_3 & \cdots & {\bf 0} & {\bf 0}\\
\vdots & \vdots & \vdots & \ddots & \vdots & \vdots\\
{\bf 0} & {\bf 0} & {\bf 0} & \cdots &{\bf P}_{G-1} & {\bf 0}\\
{\bf 0} & {\bf 0} & {\bf 0} & \cdots &{\bf 0} & {\bf P}_G\\
\end{array}
\right]
 \left[ \begin{array}{c} {{\mathbf x}_{1}} \\
{{\mathbf x}_{2}} \\
 \vdots\\
 {{\mathbf x}_{G}} \end{array} \right]+ {\mathbf n}, \label{full_jsdm}
 \end{split}
\end{equation}
where for $1\leq g \leq G$, ${\mathbf H}_{g,v}^h$  is a size $N_{d,g}\times |{\cal S}_g|$ matrix, ${\mathbf F}_{N_u, {\cal S}_G}$ is a size $|{\cal S}_g| \times N_u$ matrix, ${\mathbf P}_{g}$ is a size $|{\cal S}_g|\times |{\cal S}_g|$ matrix, and ${\mathbf x}_g$ is the group $g$ downlink symbol vector of size $|{\cal S}_g|\times 1$.
Now
due to orthogonality, we can write equivalently
\begin{equation}
\begin{split}
 {\mathbf y}_d &=   \left[ \begin{array}{c} {{\mathbf H}_{1,v}}^h\\
{{\mathbf H}_{2,v}}^h\\
 \vdots\\
 {{\mathbf H}_{G,v}}^h \end{array} \right]
\left[
\begin{array}{cccccc}
{\bf P}_1 & {\bf 0} & {\bf 0} & \cdots & {\bf 0} & {\bf 0}\\
{\bf 0} & {\bf P}_2 & {\bf 0} & \cdots & {\bf 0} & {\bf 0}\\
{\bf 0} & {\bf 0} & {\bf P}_3 & \cdots & {\bf 0} & {\bf 0}\\
\vdots & \vdots & \vdots & \ddots & \vdots & \vdots\\
{\bf 0} & {\bf 0} & {\bf 0} & \cdots &{\bf P}_{G-1} & {\bf 0}\\
{\bf 0} & {\bf 0} & {\bf 0} & \cdots &{\bf 0} & {\bf P}_G\\
\end{array}
\right]
 \left[ \begin{array}{c} {{\mathbf x}_{1}} \\
{{\mathbf x}_{2}} \\
 \vdots\\
 {{\mathbf x}_{G}} \end{array} \right]+ {\mathbf n}\\
 &= \left[ \begin{array}{c} {{\mathbf H}_{1,v}}^h{\mathbf{P_1}} {{\mathbf x}_{1}}\\
{{\mathbf H}_{2,v}}^h{\mathbf{P_2}}{{\mathbf x}_{2}}\\
 \vdots\\
 {{\mathbf H}_{G,v}}^h{\mathbf{P_G}}  {{\mathbf x}_{G}}\end{array} \right]
+ {\mathbf n}. \label{simple_jsdm}
 \end{split}
\end{equation}
\end{proof}
Since each group's precoding becomes independent of other groups, the overall downlink precoding becomes much easier and less complex for both the transmitter and the receiver. In addition, the introduction of the pre-beamformers in the form of VCM beamforming directions also simplifies the RF chains \cite{JSDM1}. The individual precoding of each group becomes now the optimization of  a $|{\cal S}_g|\times |{\cal S}_g|$ precoding matrix ${\mathbf P}_g$, as per the next theorem.

\begin{theorem}
For each group $g$ in the VCM representation, the equivalent optimum precoder, ${\mathbf P}_g$ needs to satisfy
\begin{eqnarray}
\begin{aligned}
& \underset{{\mathbf P}_g}{\normalfont \text{maximize}}
& & I({\mathbf x_{d,g}};{\mathbf y_{d,g}})\\
& {\normalfont \text{subject to}}
& &  \mathrm {tr}({\mathbf P}_g {\mathbf P}_g^h) = N_{d,g}, \\ \label{eq_MMIMO_g}
\end{aligned}
\end{eqnarray}
where the group $g$ reception model becomes
\begin{equation}
{\mathbf y}_{d,g} = {{\mathbf H}_{g,v}}^h {\mathbf P}_g {\mathbf x}_g + {\mathbf n}_g,
\end{equation}
where ${{\mathbf H}_{g,v}}^h$ is the VCM group's downlink matrix of size $N_{d,g}\times |{\cal S}_g|$, ${\mathbf y}_{d,g}$ is the group's size $N_{d,g}$ reception vector, and ${\mathbf n}_g$ is the corresponding noise.
This per group precoding problem is equivalent to a precoding problem within the original group channel model, i.e., the VCM transformation does not result in mutual information gain loss in the precoding process. In other words, employing ${\mathbf F}_{N_u, {\cal S}_g}$ as beamforming matrix per each group $g$ ($1\leq g \leq G$), is optimal from a maximization of input-output mutual information standpoint.
\end{theorem}
\begin{proof}
The only part of the theorem that needs proof is the one relating to no information loss. This is easy to prove, since the model in (\ref{full_jsdm}) relies equivalently on a $ {\mathbf F}_{N_u, {\cal S}_g}{\mathbf P}_g$ precoder and the channel is ${{\mathbf H}_{g,v}}^h{\mathbf F}_{N_u, {\cal S}_g}^h$, the optimal precoder's left singular vector matrix has to be equal to the Hermitian matrix of the right singular vector matrix of ${{\mathbf H}_{g,v}}^h{\mathbf F}_{N_u, {\cal S}_g}^h$ \cite{Xiao}. Assume that the Singular Value Decomposition (SVD) of ${\mathbf H}_{g,v}^h= {\mathbf U} {\mathbf S} {\mathbf V}^h.$ Then, it is easy to show that the right singular vector matrix of ${{\mathbf H}_{g,v}}^h{\mathbf F}_{N_u, {\cal S}_g}^h$ is equal to ${\mathbf F}_{N_u, {\cal S}_g}{\mathbf V}$, under the condition that $N_u > |{\cal S}_g|,~N_u >  N_{d,g}$, which is true in Massive MIMO. Thus, based on this theorem, the pre-beamformers applied herein are optimal for the JSDM system presented.
\end{proof}

Having put forth the premise for the model, after the basic theorems of our model are proven, we can proceed to precode with the ${\mathbf P}_g$ precoders for each group and apply PGP \cite{TE_TWC} which divides each group further into subgroups in order to simplify both the transmitter as well as the receiver complexity exponentially with small gain loss \cite{TE_TWC}.  Further, by combining PGP with the Gauss-Hermite approximation, we are able to derive the PGP optimal precoder for finite QAM constellations with constellation size $M$ easily \cite{TK_EA_QAM}.
\subsection{Uniform Planar Array (UPA) at the Base}
The concept generalizes easily to  Uniform Planar Arrays (UPA), both for arrays formed in the $z,~x$ plane as well as in the $x,~y$ plane, a UPA deployed along the $x,~y$ plane is shown in Fig. 2. The theory behind planar arrays results in a Kronecker product of two virtual channels, one channel per array dimension, as shown below. UPAs result in a three-dimensional spatial representation, thus offering higher user capacity per cell. Among the two UPA possibilities, we present the analysis for an $x,~z$ direction deployed UPA, as the analysis for an $x,~y$ direction deployed UPA is very similar.
\begin{figure}
\centering
\setcounter{figure}{1}
\includegraphics[height=3.6in,width=5.25in]{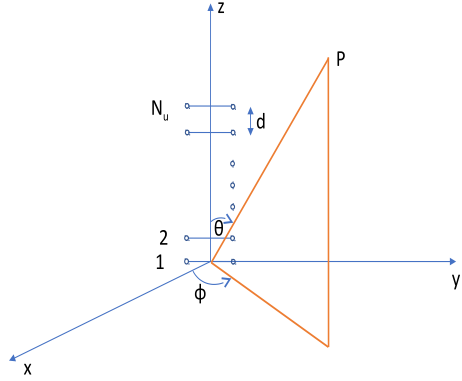}
	\caption{A UPA deployed across the $x,~y$ axes, together with the projection of a tentative transmission point on the $x,~y$ plane.}
\end{figure}
For a UPA formed on $x,~z$ directions, each group $g$'s uplink channel, ${\mathbf H}_g$, corresponds to the combination of $N_{u,x}$ uniform linear arrays deployed along the $x$ direction with $N_{u,z}$ uniform linear elements deployed in the $z$ direction. Without loss in generality, we assume that the normalized distances are the same for each direction and equal to $D$. UPAs results in a two-dimensional (matrix) antenna response matrix per user, group, and antenna expressed as
\begin{equation}
\begin{split}
{\mathbf H}_{u,g,k,n} = \sum_{l=1}^{L}\beta_{lgkn}{\mathbf a}_x(\theta_{lgkn},\phi_{lgkn}){\mathbf a}_z^t(\theta_{lgkn}), \label{eq_UPA}
\end{split}
\end{equation}
where the path gain $\beta_{lgkn}$ is as in the ULA case, $\theta_{lgkn}$ is the elevation angle for the $z$-element, same as in the ULA case, and $\phi_{lgkn}$ is the azimuth angle of user $k$'s $n$ antenna for group $g$, assumed to be a uniform r.v. in the interval $ [\phi_{g} -\Delta \phi, \phi_{g} +\Delta \phi]$. In (\ref{eq_UPA}), the two spatial vectors ${\mathbf a}_x(\theta_{lgkn},\phi_{lgkn}),~{\mathbf a}_z(\theta_{lgkn})$ are given as
\begin{equation}
        \begin{split}
        {\mathbf a}_z(\theta_{lgkn}) = &[1, \exp(-j{2\pi}D\cos(\theta_{lgkn})),\cdots, \\ &\exp(-j{2\pi D(N_{u,z}-1)}\cos(\theta_{lgkn}))]^t,
        \end{split}
        \end{equation}
         and
\begin{equation}
         \begin{split}
         {\mathbf a}_x(\theta_{lgkn},\phi_{lgkn}) = &[1, \exp(-j{2\pi}D
         \sin(\theta_{lgkn})\cos(\phi_{lgkn})),\cdots, \\ &\exp(-j{2\pi D(N_{u,x}-1)}
         \sin(\theta_{lgkn})\cos(\phi_{lgkn}))]^t,
         \end{split}
         \end{equation}
respectively.

By projecting the channel matrix ${\mathbf H}_{u,g,k,n}$ to both angular directions, i.e., on VCM for $z,~x$ directions, we get
\begin{equation}
\begin{split}
{\tilde {\mathbf H}}_{u,g,k,n} = \sum_{l=1}^{L}\beta_{lgkn}{\mathbf F}_{N_{u,x}}^h{\mathbf a}_x(\theta_{lgkn},\phi_{lgkn}){\mathbf a}_z^t(\theta_{lgkn}) {\mathbf F}_{N_{u,z}}^*. \label{eq_UPA_f}
\end{split}
\end{equation}
By taking the vector form of both sides in (\ref{eq_UPA_f}) and using identities from \cite{diff}, e.g., $\mathrm{vec}({\mathbf a}\otimes {\mathbf b}) = {\mathbf b}{\mathbf a}^t$ and $({\mathbf A}\otimes {\mathbf B})({\mathbf C}\otimes {\mathbf D}) = ({\mathbf A}{\mathbf C})\otimes({\mathbf B}{\mathbf D})$, we can write for the vector of ${\tilde {\mathbf H}}_{u,g,k,n}$, ${\tilde {\mathbf h}}_{u,g,k,n} \doteq \mathrm {vec}({\tilde {\mathbf H}}_{u,g,k,n}),$ the following equation
\begin{equation}
\begin{split}
{\tilde {\mathbf h}}_{u,g,k,n} =  \sum_{l=1}^{L}\beta_{lgkn}{\tilde {\mathbf a}}_{z,x}(\theta_{lgkn},\phi_{lgkn}),
\end{split}
\end{equation}
where
        \begin{equation}
        {\tilde {\mathbf a}}_{z,x}(\theta_{lgkn},\phi_{lgkn})= \left({\mathbf F}_{N_u,z}\otimes {\mathbf F}_{N_u,x}\right)^h \left({\mathbf a}_z(\theta_{lgkn})\otimes{\mathbf a}_x(\theta_{lgkn},\phi_{lgkn}) \right). \label{eq_UPA_concl}
        \end{equation}
The behavior in (\ref{eq_UPA_concl}) is similar with the ULA case, i.e., sparsity is achieved and different groups occupy different support sets in the angular domain.
     The expansion basis matrix now for the VCM becomes the Kronecker product of the two DFT matrices ${\mathbf F}_{N_{u,x}}$ and ${\mathbf F}_{N_{u,z}}$.
      The downlink reception model stays within the same premise, but the new Kronecker product basis is employed.
       Due to the Kronecker product, the group sparsity presents some periodicity with period equal to $N_{u,z}$.
       In other words, the reception model now becomes
       \begin{equation}
\begin{split}
 {\mathbf y}_d &=  \left[ \begin{array}{c} {{\mathbf H}_{1,v}}^h \left({\mathbf F}_{N_{u,z}}\otimes {\mathbf F}_{N_{u,x}}\right)_{{\cal S}_1}^h \\
{{\mathbf H}_{2,v}}^h\left({\mathbf F}_{N_{u,z}}\otimes {\mathbf F}_{N_{u,x}}\right)_{{\cal S}_2}^h  \\
 \vdots\\
 {{\mathbf H}_{G,v}}^h\left({\mathbf F}_{N_u,z}\otimes {\mathbf F}_{N_u,x}\right)_{{\cal S}_G}^h \end{array} \right]\\
 &\times\left[ \begin{array}{c c c c} \left({\mathbf F}_{N_{u,z}}\otimes {\mathbf F}_{N_{u,x}}\right)_{{\cal S}_1} &
\left({\mathbf F}_{N_{u,z}}\otimes {\mathbf F}_{N_{u,x}}\right)_{{\cal S}_2} & \cdots &
\left({\mathbf F}_{N_{u,z}}\otimes {\mathbf F}_{N_{u,x}}\right)_{{\cal S}_G} \end{array} \right]\\
 &\times  \left[
\begin{array}{cccccc}
{\bf P}_1 & {\bf 0} & {\bf 0} & \cdots & {\bf 0} & {\bf 0}\\
{\bf 0} & {\bf P}_2 & {\bf 0} & \cdots & {\bf 0} & {\bf 0}\\
{\bf 0} & {\bf 0} & {\bf P}_3 & \cdots & {\bf 0} & {\bf 0}\\
\vdots & \vdots & \vdots & \ddots & \vdots & \vdots\\
{\bf 0} & {\bf 0} & {\bf 0} & \cdots &{\bf P}_{G-1} & {\bf 0}\\
{\bf 0} & {\bf 0} & {\bf 0} & \cdots &{\bf 0} & {\bf P}_G\\
\end{array}
\right]+ {\mathbf n},
 \end{split}
 \end{equation}
 where the notation $({\mathbf A})_{{\cal S}_g}$ means the matrix resulting from selecting the columns of ${\mathbf A}$ that belong to ${\cal S}_g$.
                 The case of a UPA over $x,~y$ dimensions can be treated in a similar way by invoking
                 \begin{equation}
        {\tilde {\mathbf a}}_{y,x}(\theta_{lgkn},\phi_{lgkn})= \left({\mathbf F}_{N_{u,y}}\otimes {\mathbf F}_{N_{u,x}}\right)^h \left({\mathbf a}_y(\theta_{lgkn},\phi_{lgkn})\otimes{\mathbf a}_x(\theta_{lgkn},\phi_{lgkn}) \right),
        \end{equation}
        with
        \begin{equation}
        \begin{split}
        {\mathbf a}_y(\theta_{lgkn}) = &[1, \exp(-j{2\pi}D\sin(\theta_{lgkn})\sin(\phi_{lgkn})),\cdots, \\ &\exp(-j{2\pi} D(N_{u,y}-1))\sin(\theta_{lgkn})\sin(\phi_{lgkn}))]^t,
        \end{split}
        \end{equation}
        and
         \begin{equation}
         \begin{split}
         {\mathbf a}_x(\theta_{lgkn},\phi_{lgkn}) = &[1, \exp(-j{2\pi}D
         \sin(\theta_{lgkn})\cos(\phi_{lgkn})),\cdots, \\ &\exp(-j{2\pi D(N_{u,x}-1)}
         \sin(\theta_{lgkn})\cos(\phi_{lgkn}))]^t.
         \end{split}
         \end{equation}
The sparsity in the UPA case is due to the behavior of both angles, i.e., the elevation and the azimuth ones. The corresponding conditions to Lemma 1 are posted in the next lemma.
\begin{lemma}
In the UPA over $z,~x$ dimensions, when $N_u \doteq N_{u,z}N_{u,x} \gg 1$, then the significant components of the channel for group $g$, i.e., the support set ${\cal S}_g$, are found through the following two conditions
 \begin{equation}
|\cos(\theta_{lgkn})-\frac{p}{DN_{u,z}}| < \frac{1}{DN_{u,z}}, \label{resolve2}
\end{equation}
and
 \begin{equation}
|\sin(\theta_{lgkn})\cos(\phi_{lgkn})-\frac{p}{DN_{u,x}}| < \frac{1}{DN_{u,x}}. \label{resolve3}
\end{equation}
\begin{proof}
The proof stems from generalizing the condition in (\ref{resolve1}) to the geometries of the UPA array. For the $z$ direction the equation remains unchanged, while for the $x$ direction the factor $\cos(\theta_{lgkn})$ needs to be substituted by $\sin(\theta_{lgkn})\cos(\phi_{lgkn})-\frac{p}{DN_{u,x}}$. For significant factors to exist, both conditions need to be satisfied simultaneously, because the composite array factor is the product of the two individual ones. This completes the proof of the lemma.
\end{proof}
\end{lemma}
In comparison to the ULA channel case sparsity behavior though, it is important to stress that UPA channels present a repetitive, semi-periodic sparsity structure, due to the Kronecker product that exists in the vectorized form of the channel vectors. This behavior is further contrasted to the ULA one in Section V where numerical results are used to depict differences between ULA and UPA behavior with regards to sparsity in the VCM representation.
\section{THE WIDEBAND SYSTEM DESCRIPTION UNDER THE VIRTUAL CHANNEL MODEL}
        For the wideband case, we look at two possibilities: a) flat fading, and b) frequency-selective fading with OFDM. We treat both in order to facilitate a general understanding of the possibilities and different scenarios available. However, we only study the frequency-selective fading with OFDM case in our results. The presentation looks at a UPA deployed over the $z,~x$ directions. However, similar descriptions can be found for ULA and for different directions of deploying the array.

        Generalizing the previously presented scenario to wideband channels under slow fading and looking at it from $Q$ distinct frequencies adds one more dimension to the problem, i.e., one can project the frequencies to a discrete number of time components, $b$ \cite{Chen}.
The resulting channels are decomposed as a triple Kronecker product, i.e., a tensor type of product. The channel for user $k$ of group $g$ will comprise the sum of $P$ paths, each of a different delay, $\tau_l$.
Assume there are $b$ frequency slots available, starting at $0~\mathrm{Hz}$ and increasing up to $(Q-1)\Delta f$, with $\Delta f$ being the frequency bin bandwidth. We can then write the following equation for the wideband model virtual angular channel of an $x,~y$ UPA scenario
            \begin{equation}
            \begin{split}
           {\tilde {\mathbf h}}_{u,g,k,n} =
           &\frac{1}{{\sqrt L}}\sum_{l=1}^L \beta_{lgkn}
            \left({\mathbf F}_{N_u,y}\otimes {\mathbf F}_{N_u,x}\right)^h \\ &\times \left({\mathbf a}_y(\theta_{lgkn},\phi_{lgkn})\otimes{\mathbf a}_x(\theta_{lgkn},\phi_{lgkn}) \right)\delta(\tau-\tau_l). \label{eq_wide1}
           \end{split}
            \end{equation}
 Upon taking the Fourier transform with respect to $\tau$ in (\ref{eq_wide1}) we get
 \begin{equation}
\begin{split}
{\tilde {\mathbf h}}_{u,g,k,n}^{(f)}
=&\frac{1}{{\sqrt L}}\sum_{l=1}^L \beta_{lgkn}
            \left({\mathbf F}_{N_u,y}\otimes {\mathbf F}_{N_u,x}\right)^h   \\
           &\times \left({\mathbf a}_y(\theta_{lgkn},\phi_{lgkn})\otimes{\mathbf a}_x(\theta_{lgkn},\phi_{lgkn}) \right)
             {\mathbf w}_{f},
\end{split}
            \end{equation}
            where ${\mathbf w}_{f} = [ \exp(-j2\pi f \tau_1), \exp(-j2\pi f \tau_2), \cdots, \exp(-j2\pi f \tau_L)]^t$.
 Localizing the spectrum of the channel on the $Q$ frequency bins, starting at $0$ and with each bin having width $\Delta f$, we get a frequency angular domain matrix for ${\tilde {\mathbf h}}_{u,g,k,n}$ as follows
\begin{equation}
\begin{split}
{\tilde {\mathbf h}}_{u,g,k,n}^{(f)}
=&\frac{1}{{\sqrt L}}\sum_{l=1}^L \beta_{lgkn}
            \left({\mathbf F}_{N_u,z}\otimes {\mathbf F}_{N_u,x}\right)^h   \\
           &\times \left({\mathbf a}_y(\theta_{lgkn},\phi_{lgkn})\otimes{\mathbf a}_x(\theta_{lgkn},\phi_{lgkn}) \right)
             {\mathbf w}_{f,l},
\end{split}
            \end{equation}
            where
\begin{equation}
\begin{split}
 {\mathbf W}_{f} =\left[ \begin{array}{ c c c c} {1} & \exp(-j2\pi \Delta f \tau_1)  &\cdots & \exp(-j2\pi (Q-1)\Delta f \tau_1) \\
 {1} & \exp(-j2\pi \Delta f \tau_2)  &\cdots & \exp(-j2\pi (Q-1)\Delta f \tau_2) \\
 \vdots & \vdots      & \ddots & \vdots\\
{1} & \exp(-j2\pi \Delta f \tau_L) & \cdots & \exp(-j2\pi (Q-1)\Delta f \tau_L) \end{array} \right]
\end{split}
            \end{equation}
           is an $L\times Q$ Fourier transform matrix, and ${\mathbf w}_{f,l}$ is its $l$th ($1\leq l \leq L$) row.
Projecting the rows of matrix  ${\mathbf W}_{f}$ to the virtual time domain, by employing the DFT matrix of order $Q$, ${\mathbf F}_{Q}$, results in a virtual time decomposition, given as ${\mathbf w}_{f,l}{\mathbf F}_{Q}^*$, where ${\mathbf w}_{f,l}$ is the $l$th row of ${\mathbf W}_{f}$. We can then write for the entire channel representation in the virtual domain, both angular and time, the following equation
\begin{equation}
\begin{split}
{ \tilde{\mathbf H}}_{u,g,k,n}^{(t)}
=&\frac{1}{{\sqrt L}}\sum_{l=1}^L \beta_{lgkn}
            \left({\mathbf F}_{N_u,y}\otimes {\mathbf F}_{N_u,x}\right)^h \\    &\times \left({\mathbf a}_y(\theta_{lgkn},\phi_{lgkn})\otimes{\mathbf a}_x(\theta_{lgkn},\phi_{lgkn}) \right)
 {\mathbf w}_{f,l}{\mathbf F}_{Q}^* .
\end{split}
            \end{equation}
Transforming this virtual channel model matrix to a vector and using properties of Kronecker product \cite{diff} gives for the overall virtual channel model user matrix,
\begin{equation}
\begin{split}
{\tilde {\mathbf h}}_{u,g,k,n}^{(t)}
=&\frac{1}{{\sqrt L}}\sum_{l=1}^L \beta_{lgkn}
            \left({\mathbf F}_{Q}\otimes {\mathbf F}_{N_u,y}\otimes {\mathbf F}_{N_u,x}\right)^h   {\mathbf w}_{f,l}^t \otimes{\mathbf a}_y(\theta_{lgkn},\phi_{lgkn})\\
            &\otimes  {\mathbf a}_x(\theta_{lgkn},\phi_{lgkn})\\
           &= \frac{1}{{\sqrt L}}\left({\mathbf F}_{Q}\otimes {\mathbf F}_{N_u,y}\otimes {\mathbf F}_{N_u,x}\right)^h \sum_{l=1}^L \beta_{lgkn} {\mathbf w}_{f,l}^t \otimes{\mathbf a}_y(\theta_{lgkn},\phi_{lgkn})\\
            &\otimes  {\mathbf a}_x(\theta_{lgkn},\phi_{lgkn}).
\end{split}
            \end{equation}
We see that although an interesting development has occurred in this case, its applicability in the downlink precoding scenario is limited. We will see below that this changes dramatically as one moves to the frequency-selective case, due to the potential to employ OFDM on top of the JSDM model and achieve a number of additional benefits, e.g., increased capacity, easier group processing, and easier overall system deployment.

For the frequency-selective case, e.g., OFDM, the corresponding results need to be developed. Using the Tap Delay Line (TDL) model of an FS channel \cite{Proakis}, we can write for the subcarrier domain uplink channel response of the UPA\footnote{Similar results are derived for any UPA or ULA configuration within the context of this paper.} in time domain as
  \begin{equation}
\begin{split}
{ {\mathbf h}}_{u,g,k,n}^{(f)}
=&\frac{1}{{\sqrt L}}\sum_{l=1}^L \beta_{lgkn}
             {\mathbf a}_y(\theta_{lgkn},\phi_{lgkn})\\
             &\otimes{\mathbf a}_x(\theta_{lgkn},\phi_{lgkn})\delta(\tau-\frac{l-1}{B}) ,
\end{split}
            \end{equation}
   where $\delta(\cdot)$ represents the Dirac delta function, $B$ is the system bandwidth, with $B\gg B_{COH}$ and where $B_{COH}$ is the coherence bandwidth of the channel. We can then write for the frequency response of the channel
\begin{equation}
\begin{split}
{\tilde {\mathbf H}}_{u,g,k,n}^{(f)}
=&\frac{1}{{\sqrt L}}\sum_{l=1}^L \beta_{lgkn}
            \left( {\mathbf F}_{N_u,y}\otimes {\mathbf F}_{N_u,x}\right)^h   {\mathbf a}_y(\theta_{lgkn},\phi_{lgkn})\\
            &\otimes  {\mathbf a}_x(\theta_{lgkn},\phi_{lgkn}){\mathbf f}_{L,Q,l}\\
            &=\left( {\mathbf F}_{N_u,y}\otimes {\mathbf F}_{N_u,x}\right)^h
            {\mathbf M}_h  {\mathbf F}_{L,Q} , \label{eq_ch_OFDM}
\end{split}
            \end{equation}
            where ${\mathbf F}_{L,Q}$ is the last $Q-L$ row-truncated DFT matrix of order $Q$, i.e., a matrix of size $L\times Q$, ${\mathbf f}_{L,Q,l}$ is its $l$th column, $ {\mathbf M}_h$ is an $N_{u,x}N_{u,y}\times L$ matrix equal to $[{\mathbf a}_y(\theta_{lgkn},\phi_{lgkn})
            \otimes  {\mathbf a}_x(\theta_{1gkn},\phi_{1gkn}) \cdots {\mathbf a}_y(\theta_{Lgkn},\phi_{Lgkn})\otimes  {\mathbf a}_x(\theta_{Lgkn},\phi_{Lgkn})]\mathrm {diag}[\beta_{1gkn} \cdots
            \beta_{Lgkn}]$, where $\mathrm {diag}[\cdot]$ is the diagonal matrix of the vector in the brackets.
           Thus,
            ${\tilde {\mathbf H}}_{u,g,k,n}^{(f)}$ is of size $N_{u,x}N_{u,y}\times Q$, with only a few non-zero entries on each column, all of them on the same row numbers.
    The $q$th column of ${\tilde {\mathbf H}}_{u,g,k,n}^{(f)}$ is the uplink channel impulse response denoted as ${\mathbf h}_{u,g,k,n}^{(q)}$.
By recalling the fact that the spatial channel is sparse when projected to the virtual angles, exploiting the virtual channel domain representation, and after using the channel reciprocity between uplink and downlink due to TDD, for each subcarrier, we can rewrite the downlink channel of user's $k$, antenna $n$, subcarrier $q$, and group $g$ as ${(\mathbf h}_{u,g,k,n,v}^{(q)})^h$.
     We can then write for the downlink channel over all subcarriers, ${\mathbf H}_{d,g,k,n}^{(f)},$
     \begin{equation}
\begin{split}
 {\mathbf H}_{d,g,k,n}^{(f)} =  \left[ \begin{array}{c} {(\mathbf h}_{u,g,k,n,v}^{(0)})^h \\
{(\mathbf h}_{u,g,k,n,v}^{(1)})^h \\
 \vdots\\
{(\mathbf h}_{u,g,k,n,v}^{(Q-1)})^h \end{array} \right]
 \left[ \begin{array}{c c c c} \left( {\mathbf F}_{N_u,y}\otimes {\mathbf F}_{N_u,x}\right)_{{\cal S}_g}^h  \end{array} \right],
\end{split}
\end{equation}
then by stacking together all antennas for user $k$, we get
\begin{equation}
\begin{split}
 {\mathbf H}_{d,g,k}^{(f)} =  \left[ \begin{array}{c} ({\mathbf H}_{u,g,k,v}^{(0)})^h \\
({\mathbf H}_{u,g,k,v}^{(1)})^h \\
 \vdots\\
({\mathbf H}_{u,g,k,v}^{(Q-1)})^h \end{array} \right]
 \left[ \begin{array}{c c c c} \left( {\mathbf F}_{N_u,y}\otimes {\mathbf F}_{N_u,x}\right)_{{\cal S}_g}^h  \end{array} \right],
\end{split}
\end{equation}
where ${\mathbf H}_{u,g,k,v}^{(q)} = \left[{\mathbf H}_{u,g,k,1,v}^{(q)}\cdots {\mathbf H}_{u,g,k,N_{d,k^{(g)}},v}^{(q)} \right],$ a size $|{\cal S}_g|\times N_{d,k^{(g)}}$ matrix.
 Each group, $g$ ($1\leq g \leq G$) can be considered independently due to JSDM, as explained above.
We can then employ different subcarriers for different users within a group or between different groups which is explained in more detail next.

\subsection{Combined Frequency and Spatial Division and Multiplexing (CFSDM)}
In certain scenarios, it is envisaged that there is partial overlapping between adjacent groups which can lead to significant reduction in system capacity as multiple common VCMBs need to be switched off to avoid cross-group interference. Furthermore, the user co-ordination-related issues within each group might make JSDM difficult to deploy, in general. One very promising solution to mitigate both of these problems, without sacrificing the overall system capacity, is proposed herein by virtue of a novel combination of the concept of CFSDM. This idea is described below.

In CFSDM, group support sets with common VCMBs are assigned different OFDM subcarriers. In addition, users with multiple antennas within each group are also assigned different OFDM subcarriers. Finally, for users with a single antenna on the downlink, offering multiple subcarriers is the only possibility toward higher data rates. The novelty of combining JSDM  based on the VCM decomposition as proposed here and OFDM lies over the fact that it helps mitigate interference issues associated with common inter-group VCMBs as well as intra-group co-ordination. Due to the orthogonality among the subcarriers in OFDM, it is possible for two groups with common VCMBs to receive data on two different subcarriers without interference, while utilizing all the VCMBs available to them. In a similar fashion, for users within a group, assigning different subcarriers to each user makes it feasible that each user receives its data on a separate subcarrier utilizing its own receiving antennas only, thus obliterating the requirement for user co-ordination at the receiver. Specifically, let's look at a system with FS and OFDM as described in the previous subsection. Assume the system groups are as in Section I and that the OFDM component contains $Q$ orthogonal subcarriers, for some ``high enough'' number, $Q$ (e.g., $Q\geq 64$). First, let's assume that there is overlapping of the VCMBs between groups $g$ and $g'$, i.e., ${\cal S}_g\cap {\cal S}_{g'} \neq \emptyset$. The system then assigns these groups to different subcarrier groups, say ${\cal S}_{g,q}$, ~${\cal S}_{g',q'}$, which will be defined explicitly after the user subcarriers are assigned. Since there are $K_{(g)}$ users in group $g$, there is a need to assign $K_g$ subcarriers for group $g$ and $K_{g'}$ for group $g'$, if no coordination exists between users in the groups. In order for the two groups to employ all spatial capability available to them, the two groups need to avoid interference over the common VCMBs, thus in total the two groups need $K_g + K_{g'}$ different subcarriers assigned to them. Within each group, say for group $g$, user $k^{(g)}$ employing subcarrier $q_{g,k}$, there will be a PGP precoder employed in the subcarrier domain pertaining to the following receiver model
\begin{equation}
\begin{split}
 {\mathbf y}_{d,k^{(g)}}^{({q})} &=  \left[ \begin{array}{c} ({\mathbf H}_{u,k^{(g)},v}^{({q})})^h
 \end{array} \right]
 \left[ \begin{array}{c c c c} \left( {\mathbf F}_{N_u,y}\otimes {\mathbf F}_{N_u,x}\right)_{{\cal S}_g}^h  \end{array} \right]\\
 &\times \left( {\mathbf F}_{N_u,y}\otimes {\mathbf F}_{N_u,x}\right)_{{\cal S}_g} {\mathbf P}_{g,k^{(g)}}^{({q})} {\mathbf c}_{g,k}^{({q})} + {\mathbf n}_{g,k^{(g)}}^{({q})}\\
 &=\left[ \begin{array}{c} ({\mathbf H}_{u,g,k,v}^{(q)})^h
 \end{array} \right]
  {\mathbf P}_{g,k^{(g)}}^{(q)} {\mathbf c}_{g,k}^{(q)} + {\mathbf n}_{g,k}^{(q)}. \label{eq_OFDM_1}
\end{split}
\end{equation}
Now, precoding is performed on a per user and subcarrier basis, without the need for user co-operation within the group. This CFSDM approach allows for more flexible data rate allocations on a per user basis as well as helps in overcoming issues associated with spatial overlapping between groups. The following lemma also helps simplify the precoder design when the number of group antennas $N_{d,g}$ is smaller than the number of available spatial dimensions $|{\cal S}_g|$.
\begin{lemma}
When all users in a group have the same number of antennas and with $L\ll \sqrt{Q}$ and subcarrier pairs $q,~q'$ assigned within a group satisfying $|q-q'|\ll \sqrt{Q}$, then all subcarrier virtual downlink channel matrices, i.e., for all $q=1,2,\cdots, Q$, $({\mathbf H}_{u,g,k,v}^{(q)})^h$ have the same singular values. Thus, the optimal precoder over all subcarriers is the same.
\end{lemma}
\begin{proof}
We can easily rewrite (\ref{eq_OFDM_1}) by employing Kronecker matrix products as
\begin{equation}
\begin{split}
 {\mathbf y}_{d,k^{(g)}}^{({q})}
 &=\left[ \begin{array}{c} \left({\mathbf I}_{N_{k^{(g)}}}\otimes {\mathbf f_{q,L}}^h \right)({\mathbf M}_{u,g,k,v})^h
 \end{array} \right]
  {\mathbf P}_{g,k^{(g)}}^{(q)} {\mathbf c}_{g,k} + {\mathbf n}_{g,k}^{(q)},
\end{split}
\end{equation}
where $({\mathbf M}_{u,g,k,v})$ is a $N_{k^{(g)}}L\times |{\cal S}_g|$ virtual channel matrix derived from (\ref{eq_ch_OFDM}) and ${\mathbf f_{q,L}}$ represents the $q$th column of the matrix ${\mathbf F}_{L,Q}$. Now, based on the assumptions of the lemma, for any different subcarriers assigned to the group and for all $1\leq l \leq L $, we have $\exp(j2\pi\frac{(q-q')l}{Q})\approx 1$, from which we see that the matrices are approximately $\left({\mathbf I}_{N_{k^{(g)}}}\otimes {\mathbf f_{q,L}}^h \right)({\mathbf M}_{u,g,k,v})^h$ equal for all users in the group, thus they possess approximately equal singular values.
\end{proof}
Note that for a massive MIMO system a large number of $Q$ will be needed. In addition in the millimeter wavelength channels envisaged for 5th Generation (5G) cellular wireless systems, the assumption of $L\ll \sqrt{Q}$ will also be valid, since $L$ is small \cite{MMWmodel}. Thus, by assigning contiguous frequency subcarriers to different users within groups we can achieve the conditions of the above lemma.
Based on the premise of this lemma, the optimal downlink precoder in the group is the same, independently of the subcarrier employed. This is due to the fact that for CSIT optimal precoding, the optimal precoder only depends on the singular values of the channel matrix \cite{Xiao, TE_TWC}. Thus, if many subcarriers are deployed to offer higher data rates, the precoding complexity stays the same.

It is important to stress that with CFSDM, within each group, all users share the same spatial subspace, e.g., based on the same VCMBs per group. In addition, the group users share the same time domain. However, users' signals within the group are separated based on OFDM's frequency domain orthogonality. Futhermore, based on the same principle, users between overlapping VCMBs, although they share some of the spatial subspace, are orthogonal in the frequency domain, thus they do not interfere. Finally, the requirement for different subcarriers between groups is only imposed if the two groups share many VCMBs. If only a few VCMBs are shared, an alternative possibility is to switch off those common VCMBs, thus obliterating the intergroup interference. However, one can still see advantages of employing CFSDM.

\section{Numerical Results}
In this section, we present our numerical results based on ULA and UPA Massive MIMO systems with $N_u=100$ antennas at the base station. The systems employ QAM with size $M=16,~64$. We present results for both systems with and without OFDM. We have used an $L=3$ Gauss-Hermite approximation \cite{TK_EA_QAM} which results in $3^{2N_r}$ total nodes in the Gauss-Hermite approximation due to MIMO in order to facilitate results with optimal precoding in conjunction with QAM modulation. The implementation of the globally optimizing methodology is performed by employing two
backtracking line searches, one for ${\mathbf W}$ and another one for ${\boldsymbol \Sigma}_G ^2$ at each iteration, in a fashion similar to \cite{SCSI}. For the results presented, it is worth mentioning that only a few iterations (e.g., typically $<8$) are required to converge to the optimal solution results as presented in this paper. We apply the complexity reducing method of PGP \cite{TE_TWC} which offers semi-optimal results under exponentially lower transmitter and receiver complexity \cite{TE_TWC}. PGP divides the transmitting and receiving antennas into independent groups, thus achieving a much simpler detector structure while the precoder search is also dramatically reduced as well.
We divide this section in three parts, the first part looks at the VCM sparse channel representations for ULA and UPA systems, the second one examines the performance of linear precoding for Massive MIMO without OFDM, while the third one studies systems with OFDM. We use $N_{t,v}, N_{r,v}$ to denote the number of data symbol inputs, and the number of of antenna outputs, respectively, in the virtual domain. By employing PGP, one can trade in higher values of $N_{t,v}, ~N_{r,v}$ for higher overall throughput, albeit at a slightly increased complexity at the transmitter and receiver, as explained in detail in some of the examples below. Alternatively, one can employ a smaller number of $N_{t,v},~N_{r,v}$, in order to achieve higher throughput, but at significantly lower complexity. In all cases, it is stressed that the actual number of transmission and reception antennas stays the same, while all physical antennas are employed always. The details of these techniques are omitted here due to space limitation.

It is worthwhile mentioning that for precoding methods with finite inputs, two types of channels are regularly present in the literature \cite{Xiao, SCSI,Xiao2,Lozano,TK_EA_QAM,TE_TWC}: a) Type-I channels in which the precoder offers gain in the lower $\mathrm {SNR}$ regime, and b) Type II channels in which the precoder offers gain in the high $\mathrm {SNR}$ regime. Our results herein fully corroborate this type of behavior in all cases considered.
\subsection{VCM Channel Sparsity for ULA and UPA Scenarios}
First, we present results for the sparse behavior of the VCM representation in the ULA case. We randomly create 5 groups of channels as per the ULA model presented. The base ULA is deployed along the $z$ direction with $N_u=100$ elements spaced at a normalized distance $D=0.5$. There are $L=5$ paths in each channel (a smaller number of $L$ results in sparser representations). The elevation angles for groups $G_1,~G_2,~G_3,~G_4,~G_5$ are at $5^\circ,~33^\circ,~61^\circ,~89^\circ,~\text{and}~117^\circ$, respectively. In addition, the groups possess $16,~2,~4,~4,~\text{and}~6$ antennas, respectively. The angular spread for all groups is taken to be $\pm~4^\circ$ around the elevation angle of each group. The channels are projected to the VCM space, then only components greater than 1 in absolute square power are selected. In all cases considered, this selection process results in more than $94\%$ of the total power of each channel selected. The corresponding, non-overlapping support sets are as follows (the numbers of each set correspond to the numbered components of the VCM representation vector, i.e., the significant VCMBs):\newline
${\cal S}_1 = [ 56 ,~   57,~    58,~    59,~    60,~    61,~    62,~    63,~    64,~    65 ]$,\\ ${\cal S}_2 = [ 38,~    39,~    40,~    41,~    42,~    43,~    44]$,\\
${\cal S}_3 = [  27,~    28,~    29,~    30,~    31,~    32,~    33,~    34]$,\\ ${\cal S}_4 = [   1,~     2,~     3,~     4,~    5,~     6,~     7,~    99,~   100]$,\\ ${\cal S}_5 = [   70,~    71,~    72,~    73,~    74,~    75,~    76,~    77,~    78,~    79]$.\newline We observe that a ULA allows for easy sparse non-overlapping support sets for multiple groups.

Next we present similar results for a UPA array along the $x,~y$ directions. In this example, there are 8 groups, $G_1$ through $G_8$, formed. The normalized distance between successive elements in both directions is $D=0.6$, while the number of elements on each direction is equal to 10, i.e., $N_{u,x}=N_{u,y}=10$. There are a total of $16,~1,~4,~4,~6,~6,~\text{and}~1$ antennas available for each group. The angle spread per dimension is $\pm 2^\circ$, while $L=2$. The corresponding VCMBs per group are as follows:\newline ${\cal S}_1 =[ 1,~ 2,~ 3,~    4,~     5,~     6,~     7,~     8,~     9,~    10,~    11,~    12,~    20,~    21,~    31,~    41,~    51,~    61,~    71,~    81,~    91]$, \\${\cal S}_2 =[ 12,~    13,~    14,~    15]$, \\${\cal S}_3 =[  2,~     3,~     4,~     5,~     6,~    11,~    12,~    13,~    14,~    15,~    16,~    17,~    23,~    24,~    34,~    44,~    54,~    64,~    74,~    84,~    93,~    94,~
]$, \\${\cal S}_4 =[   3,~     4,~    5,~     6,~    14,~    15,~    24,~    34,~    64,~    74,~    75,~    84,~    85,~    91,~    92,~    93,~    94,~    95,~    96,~    97,~    98,~    99,~   100]$, \\${\cal S}_5 =[ 74,~    83,~    84,~    85,~    94]$, \\${\cal S}_6 = [73,~    83 ]$,\\ ${\cal S}_7 =[ 1,~     2,~    11,~    21,~    61,~    71,~    81,~    82,~    91,~    92,~    93,~    94,~    95,~   96,~    97,~    98,~    99,~   100]$, \\${\cal S}_8 =[1,~     2,~     3,~     4,~     8,~     9,~    10,~    11,~    20,~    21,~    31,~    41,~    51,~    61,~    71,~    81,~    91,~    92,~    99,~   100 ]$. \newline It is easy to see that UPA deployments offer more VCMBs per group, however at a cost to orthogonality. In addition, UPAs offer better resolution compared to ULAs, thus they could in principle offer higher capacity. An additional benefit of a UPA is the fact that one gets more VCMBs per group thus the resulting throughput with precoding is higher. Due to the significant overlapping between different group VCMBs, there are two options when UPAs are selected for higher capacity: a) Release common VCMBs, i.e., leave the common VCMBs between groups unused, however at the expense of performance, or, b) Employ OFDM in parallel to the JSDM in the system. The latter approach can offer very high capacity due to its capability to mitigate overlapping in spatial domain while at the same time it takes advantage of orthogonality between non-overlapping VCMBs. Both approaches are explained in more detail below.

\subsection{Precoding Results without OFDM}
As a first example, we present results for a ULA with 5 groups formed, shown as $G_1,~G_2,\cdots,G_5$, respectively. Groups $1,~2,~3,~4,~5$ occupy the following groups of non-overlapping, i.e., disjoint VCMBs \newline${\cal S}_1=[57,~58,~59,~60,~61,~62,~63,~64]$, \\${\cal S}_2=[39,~41,~42,43,~44,~45,~46,~47]$, \\${\cal S}_3=[25,~26,~27,~28,~29,~30,~31,~32,~33,34,~35]$, \\${\cal S}_4=[1,~2,~3,~4,~5,~6,~98,~99,100]$, and \\ ${\cal S}_5=[68,~69,~70,~71,~72,~73,~74,~75,~76,~77]$,\newline respectively. The groups include $4,~2,~4,~4,~6$ antennas at the User Equipment (UE), respectively. In the non-OFDM case, users within groups need to co-ordinate their downlink. Thus, the number of users within the group becomes irrelevant and only the number of antennas becomes essential.
In Fig. 3 we present results for $G_4$. We observe that high gains in throughput are available for low $\mathrm {SNR}$, i.e., a Type-I channel behavior. For example, at $\mathsf{SNR}_b=-7~dB$ there is an 33\% throughput increase by using PGP over the no precoding case. In addition, there is an precoding gain of $4-5~dB$ over the low $\mathrm {SNR}$ regime. As far as complexity is concerned, based on the analysis of \cite{TK_EA_QAM}, the PGP precoding example presented with $N_{t,v}=6$ require a complexity (both at the transmitter and receiver) on the order of $3M^4$, while the no precoding example requires a complexity at the receiver on the order of $M^{18}$, thus PGP needs $(1/3)M^{14}$ less complexity. For the $N_{t,v}=8$ case the complexity reduction with PGP over the no PGP case becomes $(1/4)M^{14}$. Thus, we see that PGP helps keep the UE complexity low, while it gives significant gains in throughput and $\mathrm {SNR}$.
\begin{figure}[h]
\centering
\setcounter{figure}{2}
\includegraphics[height=2.4in,width=3.5in]{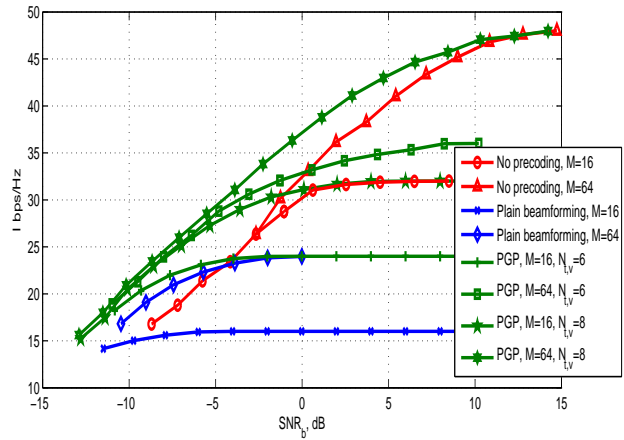}
	\caption{$I({\mathbf x};{\mathbf y})$ results for PGP, plain beamforming, and no-precoding cases for the channel in $G_4$ in conjunction with QAM $M=16,~64$ modulation.}
\end{figure}
In Fig. 4 we present results for $G_5$. Here, we observe high gains in throughput in  high $\mathrm {SNR}$ regime. Here we employ $N_{t,v}=6$. We observe that this is a Type-II channel behavior. At $\mathsf{SNR}_b > 0$, the no precoding case throughput saturates at $40~bps/Hz$. However, with PGP we get significantly higher throughput, e.g., at $\mathsf{SNR}_b = 10~dB$ the throughput is $48~bps/Hz$. Further, it takes PGP $(1/6)M^{16}$ less UE complexity  than the no precoding one in order to achieve this additional throughput at the UE.
\begin{figure}[h]
\centering
\setcounter{figure}{3}
\includegraphics[height=2.4in,width=3.5in]{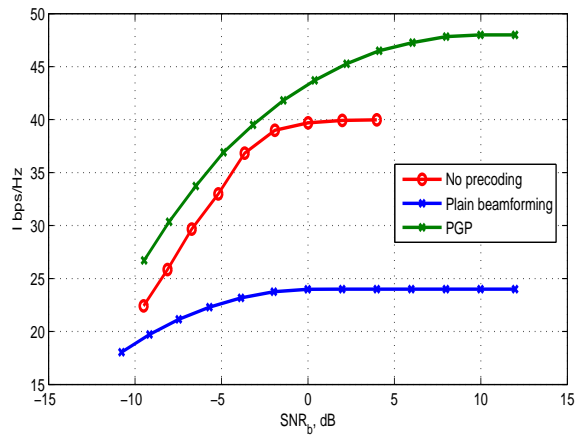}
	\caption{$I({\mathbf x};{\mathbf y})$ results for PGP, plain beamforming, and no-precoding cases for the channel in $G_5$ in conjunction with QAM $M=16$ modulation.}
\end{figure}

For a UPA along the $z,~x$ directions, with $N_{u,z}=N_{u,x}=10$, $D=0.6$, we get 8 groups with the following VCMBs:   \newline${\cal S}_1=[1,     2,     3,    10,    11,    12,    21,    31,    41,    71,    81,    91],$ \newline
${\cal S}_2=[ 3,     4,    11,    12,    13,    14,    15,    16,    17,    20,    23,    33,    93]$, \newline
 ${\cal S}_3=[ 3,     4,    14,    94,
]$, \newline
 ${\cal S}_4=[ 4,    14,    24,    34,    44,    54,    64,    74,    83,    84,    85,    93 ,   94,    95]$,\newline
  ${\cal S}_5=[53,    63,    72,    73,    74 ,   83,    93,]$,\newline
  ${\cal S}_6=[ 62,    72]$, \newline
   ${\cal S}_7=[ 1 ,   11,    21 ,   61,    71,    81,    91,    92 ,   93,    99,   100,
]$,\newline
and ${\cal S}_8=[1,     2,     3,     4 ,    5,     6,     7 ,    8 ,    9,    10,    11,    21,    81,    91 ,  100
]$. \newline The corresponding number of each group UE antennas is $4,~2,~4,~4,~6,~1,~6,~\text{and}~8$, respectively. We see that partial overlapping exists between different groups VCMBs. Without OFDM, we need to leave the common VCMBs unused to avoid primary interference between groups. We thus end up with the following revised sets:\newline  ${\cal S}_1=[31,~41]$, \\${\cal S}_2 =[13,~15,~16,~17]$, \\${\cal S}_3 =[94]$, \\${\cal S}_4 =[64,~84,~85,~95]$, \\${\cal S}_5=[53,~63,~73]$, \\${\cal S}_6 =[62]$, \\${\cal S}_7 =[61,~92,~93,~99]$, and \\${\cal S}_8=\emptyset$.\newline In Fig. 5 we present results on the $G_1$ downlink precoding where we have applied PGP with two additional ``ficticious'' inputs, similar to \cite{TK_EA_QAM} and see dramatic improvements on downlink throughput. We see the dramatic impact of VCMB overlapping in the case of UPA. Notice that the complexity involved in the PGP is two times higher than the one on the no precoding case, due to $N_{t,v}=4$ ``ficticious'' antennas being introduced, while the incurred loss in $G_1$ due to the reduction on the number of useful VCMBs is highly mitigated. This example is a Type-II channel behavior in which PGP achieves double the throughput in high $\mathrm {SNR}$, while the corresponding UE complexity is two times higher than the no precoding one, since $N_{t,v}=4> N_t$.
\begin{figure}[h]
\centering
\setcounter{figure}{4}
\includegraphics[height=2.4in,width=3.5in]{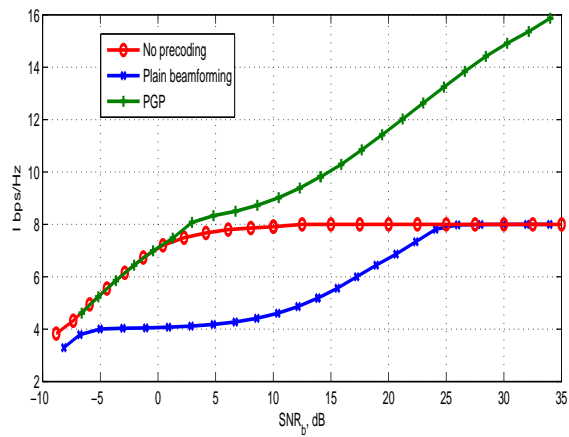}
	\caption{$I({\mathbf x};{\mathbf y})$ results for PGP, plain beamforming, and no-precoding cases for the channel in $G_1$ in conjunction with QAM $M=16$ modulation.}
\end{figure}
For the same system, in $G_2$ we get the results presented in Fig. 6. For the PGP and plain beamforming cases we show results for both $M=16,~64$. The PGP and plain beamforming results use $N_{t,v}=N_t=4$ ``ficticious'' antennas each, the same number as the no precoding case. We observe $\mathrm {SNR}$ and thoughput gains in low $\mathrm {SNR}$. For example an $\mathrm {SNR}$ gain higher than 8 dB with PGP in the $\mathsf{SNR_b}$ over the no precoding case in low $\mathrm {SNR}$, while the incurred UE receiver complexity with PGP is $(1/2)M^4$ times lower than the no precoding case.
\begin{figure}[h]
\centering
\setcounter{figure}{5}
\includegraphics[height=2.4in,width=3.5in]{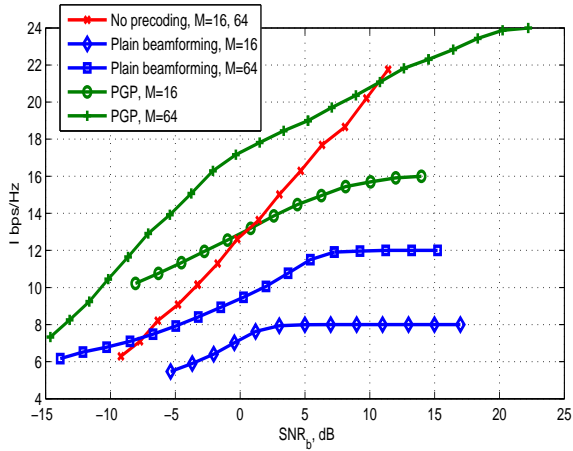}
	\caption{$I({\mathbf x};{\mathbf y})$ results for PGP, plain beamforming, and no-precoding cases for the channel in $G_2$ in conjunction with QAM $M=16,~64$ modulation.}
\end{figure}
\subsection{Precoding Results with OFDM}
We next present results with OFDM. For a UPA deployed over the $x,~y$ directions, with $N_{u,x}=N_{u,y}=10$, an OFDM system with $Q=64$ subcarriers, we get 3 groups with the following VCMB's. $G_1$ has \newline ${\cal S}_1 = [1, ~2,    ~10,    ~11,    ~21,   ~ 81,    ~91
 ]$,\newline $G_2$ has ${\cal S}_2 =[ 11,    ~12,    ~13,   ~ 14,   ~15,    ~16,   ~ 17, ~   18,~    19,~    20]$, \newline and $G_3$ has ${\cal S}_3 = [3,~     4,~     5,~    12,~    13,~    14,~    15,~    24,~    34,~    44,~    54,~    64,~    74,~    84,~    94]$. \newline $G_1$ comprises 2 users with two antennas each, $G_2$ and $G_3$ comprise 2 users with 4 antennas each. There is VCMB overlapping between the groups, however by employing CFSDM we can avoid the interference coming from overlapping VCMBs. In addition, by employing different subcarriers between the different users in each group in CFSDM, we can avoid joint decoding within the group level, i.e., the users decode their data totally independently. In this particular example we envisaged employing in total 6 OFDM subcarriers, 2 per group for all 3 groups. In Fig. 7 and Fig. 8 we present results for user 1, user 2 of $G_1$, respectively. In both cases we see Type-I channel behavior. In this example, both users employ 2 receiving antennas. By virtue of CFSDM, the downlink can employ all VCMBs for both users, i.e., no need to partition the VCMB set. The example here applies 4 downlink pre-beamformers per user and in the PGP results we use 2 groups of size $4\times 4$ each, by extending the receiving antennas to 4, using 2 ``ficticious'' antennas, i.e., $N_{t,v}=4$ in a fashion similar to \cite{TK_EA_QAM}. Futhermore, a revised, improved version of plain beamforming is used in which, only inputs with non-zero associated singular values are employed. We call this form of plain beamforming Singular Value Aware Plain Beamforming (SVAPB). We see very high throughput and $\mathrm {SNR}$ gains offered by PGP over the no precoding in low $\mathrm {SNR}$, and the plain beamforming case, over all shown $\mathrm {SNR}$, respectively, although the latter performs better than standard beamforming due to SVAPB. In Fig. 7 we show at $\mathsf{SNR}_b=5~dB$ more than 3 times better throughput with PGP than the no precoding case, while for a quite wide range of $\mathsf{SNR}_b$ we see gains on the order of $8~dB$ in $\mathrm {SNR}$. The corresponding complexity with PGP is $(1/2)M^{10}$ times lower than the no precoding one. In Fig. 8 we observe a gain in throughput of $33\%$ at $\mathsf{SNR}_b = 15~dB$, while the $\mathrm {SNR}$ gain is on the order of $5~dB$. The corresponding complexity with PGP is same with the one in Fig. 7, i.e.,  $(1/2)M^{10}$ lower than the no precoding one.
 \begin{figure}[h]
\centering
\setcounter{figure}{6}
\includegraphics[height=2.4in,width=3.5in]{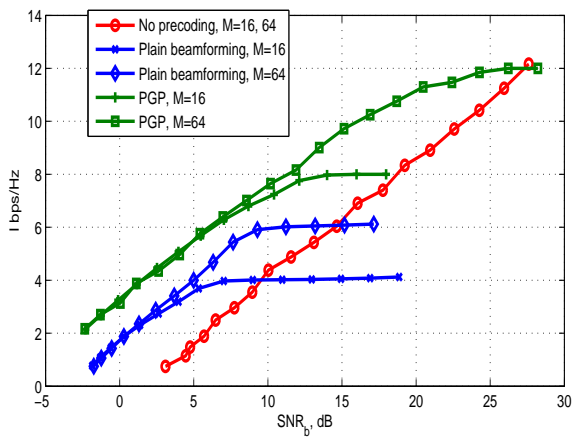}
	\caption{$I({\mathbf x};{\mathbf y})$ results for PGP, plain beamforming, and no-precoding cases for the channel in group $G_1$, user 1 in conjunction with QAM $M=16,~64$ modulation and CFSDM.}
\end{figure}

\begin{figure}[h]
\centering
\setcounter{figure}{7}
\includegraphics[height=2.4in,width=3.5in]{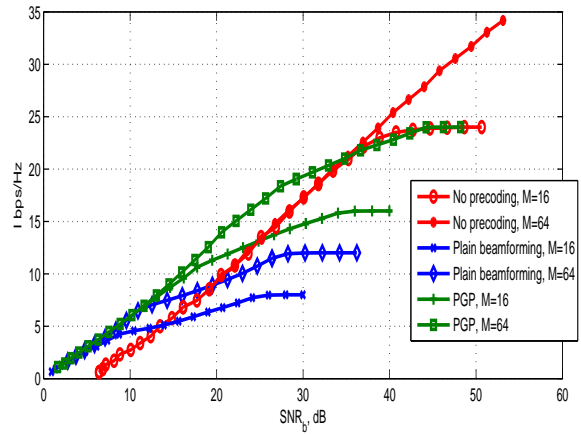}
	\caption{$I({\mathbf x};{\mathbf y})$ results for PGP, plain beamforming, and no-precoding cases for the channel in $G_1$, user 2 in conjunction with QAM $M=16,~64$ modulation and CSFDM.}
\end{figure}

\section{Conclusions}
In this paper, a novel methodology for Massive MIMO systems is presented, allowing for optimal downlink linear precoding with finite-alphabet inputs, e.g., QAM and multiple antennas per user. The methodology is based on a sparse VMC decomposition of the downlink channels, which then allows for orthogonality between different user groups, due to non-overlapping sets of VCMBs. The methodology is applied in systems with or without OFDM and for ULA and UPA antenna configurations. By employing the PGP technique to the proposed system, we show very high gains are available on the downlink. However, in the non-OFDM deployment, the users in each group need to co-ordinate their detection processes in order to achieve precoding gains. When OFDM is available, there is more flexibility in system design. For example, users in the group can be assigned different subcarriers, thus ameliorating the need for intra-group detection coordination. In addition, in cases of partial overlapping of the available VCMB sets, by employing separate subcarriers, as in OFDM, the interfering groups can become completely orthogonal, thus fully mitigating the inter-group interference due to partial VCMB overlap. The novel combination of OFDM with the VCM-based JSDM system presented is called Combined Frequency and Spatial Division and Multiplexing (CFSDM) and offers additional advantages, such as high throughput to users with single antenna and it also obliterates the need for intragroup user decoding coordination. Our numerical results show high gains, e.g., typically higher than $60\%$ and in some cases as high as $200\%$ in throughput while the incurred precoding complexity is exponentially lower at both the transmitter and receiver sites.
\bibliographystyle{IEEEtran}

\input{ARXIV.bbl}

\end{document}

%% file: ARXIV.bbl

%% file: CFSDM_ARXIV_3.bbl
\begin{thebibliography}{10}
\providecommand{\url}[1]{#1}
\csname url@samestyle\endcsname
\providecommand{\newblock}{\relax}
\providecommand{\bibinfo}[2]{#2}
\providecommand{\BIBentrySTDinterwordspacing}{\spaceskip=0pt\relax}
\providecommand{\BIBentryALTinterwordstretchfactor}{4}
\providecommand{\BIBentryALTinterwordspacing}{\spaceskip=\fontdimen2\font plus
\BIBentryALTinterwordstretchfactor\fontdimen3\font minus
  \fontdimen4\font\relax}
\providecommand{\BIBforeignlanguage}[2]{{%
\expandafter\ifx\csname l@#1\endcsname\relax
\typeout{** WARNING: IEEEtran.bst: No hyphenation pattern has been}%
\typeout{** loaded for the language `#1'. Using the pattern for}%
\typeout{** the default language instead.}%
\else
\language=\csname l@#1\endcsname
\fi
#2}}
\providecommand{\BIBdecl}{\relax}
\BIBdecl

\bibitem{Marzetta1}
T.~Marzetta, ``{N}oncooperative {C}ellular {W}ireless with {U}nlimited
  {N}umbers of {B}ase {S}tation {A}ntennas,'' \emph{IEEE Transactions on
  Wireless Communications}, vol.~9, pp. 3590--3600, November 2010.

\bibitem{Marzetta2}
J.~Jose, A.~Ashikhmin, T.~Marzetta, and S.~Vishwanath, ``{P}ilot
  {C}ontamination and {P}recoding in {M}ulti-{C}ell {TDD} {S}ystems,''
  \emph{IEEE Transactions on Wireless Communications}, vol.~10, pp. 2640--2651,
  August 2011.

\bibitem{Marzetta3}
H.~Ngo, E.~Larsson, and T.~Marzetta, ``{E}nergy and {S}pectral {E}fficiency of
  {V}ery {L}arge {M}ultiuser {MIMO} {S}ystems,'' \emph{IEEE Transactions on
  Communications}, vol.~61, pp. 1436--1449, April 2013.

\bibitem{JSDM1}
J.~Nam, J.~Y. Ahn, A.~Adhikary, and G.~Caire, ``{J}oint spatial division and
  multiplexing: The large-scale array regime,'' \emph{IEEE Trans. Inf. Theory},
  vol.~59, pp. 6441--6463, October 2012.

\bibitem{Xiao}
C.~Xiao, Y.~Zheng, and Z.~Ding, ``{G}lobally {O}ptimal {L}inear {P}recoders for
  {F}inite {A}lphabet {S}ignals {O}ver {C}omplex {V}ector {G}aussian
  {C}hannels,'' \emph{IEEE Transactions on Signal Processing}, vol.~59, pp.
  3301--3314, July 2011.

\bibitem{Lamarca}
M.~Lamarca, ``{L}inear {P}recoding for {M}utual {I}nformation {M}aximization in
  {MIMO} {S}ystems,'' in \emph{Proceedings International Symposium of Wireless
  Communication Systems 2009}, 2009, pp. 26--30.

\bibitem{Verdu1}
F.~Perez-Cruz, M.~Rodriguez, and S.~Verdu, ``{MIMO} {G}aussian {C}hannels with
  {A}rbitrary {I}nputs: {O}ptimal {P}recoding and {P}ower {A}llocation,''
  \emph{IEEE Transactions on Information Theory}, vol.~56, pp. 1070--1086,
  March 2010.

\bibitem{Max}
M.~Girnyk, M.~Vehkapera, and L.~K. Rasmussen, ``{L}arge {S}ystem {A}nalysis of
  {C}orrelated {MIMO} {M}ultiple {A}ccess {C}hannels with {A}rbitrary
  {S}ignaling in the {P}resence of {I}nterference,'' \emph{IEEE Transactions on
  Wireless Communications}, vol.~4, pp. 2060--2073, April 2014.

\bibitem{Khan}
D.~Shiu, G.~Foschini, M.~Gans, and J.~Kahn, ``{F}ading {C}orrelation and {I}ts
  {E}ffect on the {C}apacity of {M}ultielement {A}ntenna {S}ystems,''
  \emph{IEEE Transactions on Communications}, vol.~48, pp. 502--513, March
  2000.

\bibitem{Weich}
W.~Weichselberger, M.~Herdin, H.~Ozcelik, and E.~Bonek, ``A {S}tochastic {MIMO}
  {C}hannel {M}odel with {J}oint {C}orrelation of {B}oth {L}inks,'' \emph{IEEE
  Transactions on Wireless Communications}, vol.~5, pp. 90--100, January 2006.

\bibitem{SCSI}
Y.~Wu, C.-K. Wen, C.~Xiao, X.~Gao, and R.~Schober, ``Linear {P}recoding for the
  {MIMO} {M}ultiple {A}ccess {C}hannel {W}ith {F}inite {A}lphabet {I}nputs and
  {S}tatistical {CSI},'' \emph{IEEE Transactions on Wireless Communications},
  pp. 983--997, February 2015.

\bibitem{Lozano}
Y.~Wu, C.-K. Wen, D.~Ng, R.~Schober, and A.~Lozano, ``{L}ow-{C}omplexity {MIMO}
  {P}recoding with {D}iscrete {S}ignals and {S}tatistical {CSI},'' in
  \emph{Proceedings ICC}, 2016.

\bibitem{TK_EA_QAM}
T.~Ketseoglou and E.~Ayanoglu, ``Linear {P}recoding {G}ain for {L}arge {MIMO}
  {C}onfigurations with {QAM} and {R}educed {C}omplexity,'' \emph{IEEE
  Transactions on Communications}, vol.~64, pp. 4196--4208, October 2016.

\bibitem{TE_TWC}
------, ``Linear {P}recoding for {MIMO} with {LDPC} {C}oding and {R}educed
  {C}omplexity,'' \emph{IEEE Transactions on Wireless Communications}, pp.
  2192--2204, April 2015.

\bibitem{TK_EA_GLOBECOM}
------, ``Linear {P}recoding {G}ain for {L}arge {MIMO} {C}onfigurations with
  {QAM} and {R}educed {C}omplexity,'' in \emph{2016 IEEE Global Communications
  Conference (GLOBECOM)}, 2016.

\bibitem{Sayeed}
A.~M. Sayeed, ``Deconstructing {M}ultiantenna {F}ading {C}hannels,'' \emph{IEEE
  Transactions on Signal Processing}, vol.~50, no.~10, pp. 2563--2579, October
  2002.

\bibitem{BEM}
H.~Xie, F.~Gao, S.~Zhang, and S.~Jin, ``Spatial-{T}emporal {BEM} and {C}hannel
  {E}stimation {S}trategy for {M}assive {MIMO} {T}ime-{V}arying {S}ystems,'' in
  \emph{2016 IEEE Global Communications Conference (GLOBECOM)}, December 2016,
  pp. 1--6.

\bibitem{diff}
A.~Hj$\emptyset$rugnes, \emph{Complex-Valued Matrix Derivatives With
  Applications in Signal Processing and Communications}.\hskip 1em plus 0.5em
  minus 0.4em\relax Cambridge, UK: Cambridge University Press, 2011.

\bibitem{Chen}
Z.~Chen and C.~Yang, ``Pilot {D}econtamination in {W}ideband {M}assive {MIMO}
  {S}ystems by {E}xploiting {C}hannel {S}parsity,'' \emph{IEEE Transactions on
  Wireless Communications}, vol.~15, no.~7, pp. 5087--5100, July 2016.

\bibitem{Proakis}
J.~Proakis, \emph{Digital {C}ommunications}.\hskip 1em plus 0.5em minus
  0.4em\relax New York: McGraw-Hill, 2001.

\bibitem{MMWmodel}
Z.~Pi and F.~Khan, ``An {I}ntroduction to {M}illimeter-wave {M}obile
  {B}roadband {S}ystems,'' \emph{IEEE Communications Magazine}, vol.~49, no.~6,
  pp. 101--107, June 2011.

\end{thebibliography}
